%% file: paper.tex
\documentclass[11pt]{article}
\usepackage{amsmath}
\usepackage{amsthm}
\usepackage[hmargin=1in,vmargin=1in]{geometry}
\newtheorem{lemma}{Lemma}
\newtheorem{theorem}{Theorem}
\newtheorem{corollary}{Corollary}

\newtheorem{fact}{Fact}

\def\id#1{\ensuremath{\mathit{#1}}}

\def\idbf#1{\ensuremath{\mathbf{#1}}}

\providecommand{\etal}{~et al.}

\newcommand{\zero}{\idbf{0}}
\newcommand{\one}{\idbf{1}}
\newcommand{\A}{\texttt{A}}
\newcommand{\B}{\texttt{B}}
\newcommand{\I}{\texttt{I}}
\newcommand{\D}{\texttt{D}}
\newcommand{\R}{\texttt{R}}
\newcommand{\V}{\texttt{V}}
\newcommand{\select}{\id{select}}
\newcommand{\search}{\id{search}}
\newcommand{\rank}{\id{rank}}
\newcommand{\rankzero}{\id{rank}_\zero}
\newcommand{\rankone}{\id{rank}_\one}
\newcommand{\selectzero}{\id{select}_\zero}
\newcommand{\selectone}{\id{select}_\one}
\newcommand{\rankb}{\id{rank}_b}
\newcommand{\selectb}{\id{select}_b}
\newcommand{\flip}{\id{flip}}

\newcommand{\ins}{\id{insert}}
\newcommand{\delete}{\id{delete}}
\newcommand{\del}{\id{delete}}

\begin{document}

\title{Near-Optimal Online Multiselection\\ in Internal and External Memory}
\author{J\'er\'emy Barbay\thanks{Supported by Project Regular Fondecyt number 1120054.}\\
Departamento de Ciencias de la Computaci{\'o}n (DCC)\\
Universidad de Chile\\\texttt{jbarbay@dcc.uchile.cl} 
\and Ankur Gupta\thanks{Supported in part by the Butler Holcomb Awards grant.}\\
Department of Computer Science and Software Engineering\\
Butler University\\\texttt{agupta@butler.edu} 
\and S.~Srinivasa Rao\thanks{Supported by Basic Science Research Program through the National Research Foundation of Korea (NRF) funded by the Ministry of Education, Science and Technology (Grant number 2012-0008241).}\\
School of Computer Science and Engineering\\
Seoul National University\\\texttt{ssrao@cse.snu.ac.kr} 
\and Jon Sorenson\footnotemark[2]\\
Department of Computer Science and Software Engineering\\
Butler University\\\texttt{jsorenso@butler.edu}}

\date{}
\maketitle

\begin{abstract}
We introduce an online version of the \emph{multiselection} problem, in which $q$ selection queries are requested on an unsorted array of $n$~elements. We provide the first online algorithm that is $1$-competitive with Kaligosi\etal [ICALP 2005] in terms of comparison complexity. Our algorithm also supports online \emph{search} queries efficiently.

We then extend our algorithm to the dynamic setting, while retaining online functionality, by supporting arbitrary \emph{insertions} and \emph{deletions} on the array. Assuming that the insertion of an element is immediately preceded by a search for that element, we show that our dynamic online algorithm performs an optimal number of comparisons, up to lower order terms and an additive~$O(n)$ term.

For the external memory model, we describe the first online multiselection algorithm that is $O(1)$-competitive. This result improves upon the work of Sibeyn [Journal of Algorithms 2006] when $q > m$, where $m$ is the number of blocks that can be stored in main memory. We also extend it to support searches, insertions, and deletions of elements efficiently.
\end{abstract}

\pagebreak

  \input{outline}

\pagebreak

  \bibliographystyle{alpha}
  \bibliography{bibabbrv,selsort}

\appendix
\include{appendix}

\end{document}

%% file: outline.tex
\renewcommand{\B}{\mathcal{B}}

\section{Introduction}
\label{section:intro}

\input{introduction}

\section{A Simple Online Algorithm}
\label{section:algorithm}

\input{algorithm}

\section{A Lemma on Sorting Entropy}
\label{section:entropy}

\input{entropy}

\section{Optimal Online Multiselection}
\label{section:optimal}

\input{optimal}

\section{Optimal Online Dynamic Multiselection}
\label{section:dynamic}

\input{dynamic}

\section{External Online Multiselection}
\label{section:external}

\input{EXT}

%% file: introduction.tex
The \emph{multiselection} problem asks for the elements of rank~$Q={q_1,q_2,\dots, q_q}$ on an unsorted array~$\A$ drawn from an ordered universe of elements. We define $\B(S_q)$ as the information-theoretic lower bound on the number of comparisons needed to answer $q$ queries, where $S_q = {s_i}$ denotes the queries ordered by rank. We define $\Delta_i = s_{i+1} - s_i$, where~$s_0 = 0$ and $s_{q+1} = n$. Then, 

\begin{displaymath}
\B(S_q) = \log n! - \sum_{i=0}^q \log \left(\Delta_i!\right) = \sum_{i=0}^q \Delta_i \log \frac{n}{\Delta_i} - O(n).\footnote{We use the notation $\log_b a$ to refer to the base~$b$ logarithm of~$a$. By default, we let $b=2$. We also define $\ln b$ as the base~$e$ logarithm of $b$.}
\end{displaymath} 

Several papers have analyzed this problem carefully. Dobkin and Munro~\cite{dobkin-multiselect} gave a deterministic bound using $3\B(S_q) + O(n)$ comparisons. Prodinger~\cite{Prodinger95} proved the expected comparisons with random pivoting is~$2\B(S_q)\ln 2 + O(n)$. Most recently, Kaligosi\etal~\cite{kaligosi-multiselect} showed a randomized algorithm taking $\B(S_q) + O(n)$ expected comparisons, along with a deterministic algorithm taking $\B(S_q) + o(\B(S_q) + O(n)$ comparisons. Jim\'enez and Mart\'inez~\cite{JimenezM10} later improved the number of comparisons in the expected case to $\B(S_q) + n + o(n)$. Most recently, Cardinal\etal~\cite{cardinal-partial-order-production} generalized the problem to a \emph{partial order production}, of which multiselection is a special case. Cardinal\etal~use the multiselection algorithm as a subroutine after an initial preprocessing phase.

Kaligosi\etal~\cite{kaligosi-multiselect} provide an elegant result in the deterministic case based on tying the number of comparisons required for merging two sorted sequences to the information content of those sequences. This simple observation drives an approach where manipulating these runs to both find pivots that are ``good enough'' and partition with near-optimal comparisons. 
The weakness of the approaches in internal memory is that they must know all of the queries a priori.

In external memory, Sibeyn~\cite{sibeyn:external-selection} solves multiselection using $n + nq/m^{1-\epsilon}$ I/Os, where $\epsilon$ is any positive constant. The first term comes form creating a static index structure using $n$ I/Os, and the reminder comes from the $q$ searches in that index. In addition, his results also require the condition that $\log N = O(B)$. When $q = m$, Sibeyn's multiselection algorithm requires $O(n m^\epsilon)$ I/Os, whereas the optimum is $\Theta(n)$ I/Os. In fact his bounds are $\omega(\B_m(S_q))$, for any $q \ge m$, where
$B_m(S_q))$ is the lower bound on the number of I/Os required (see Section~\ref{sec:lower-bound-mult} for the definition).

\subsection{Our Results}
For the \emph{multiselection} problem in internal memory, we describe the first online algorithm that
supports a set $Q$ of $q$ selection, search, insert, and delete operations, of which $q'$ are search, insert, and delete, using $\B(S_q) + o(\B(S_q) + O(n + q'\log n)$ comparisons.\footnote{For the dynamic result, we assume that the insertion of an element is immediately preceded by a search for that element. In that case, we show that our dynamic online algorithm performs an optimal number of comparisons, up to lower order terms and an additive~$O(n)$ term.} Thus our algorithm is $1$-competitive with the offline algorithm of Kaligosi\etal~\cite{kaligosi-multiselect} in terms of comparison complexity. We also show a randomized result achieving $1$-competitive behavior with respect to Kaligosi\etal~\cite{kaligosi-multiselect}, while only using $O((\log(n))^{O(1)})$ sampled elements instead of $O(n^{3/4})$.

For the external memory model, we describe an online multiselection algorithm that supports a set $Q$ of $q$ selection queries on an unsorted array stored on disk in $n$ blocks, using $O(\B_m(S_q)) + O(n)$ I/Os, where $B_m(S_q)$ is a lower bound on the number of I/Os required to support the given queries. This result improves upon the work of Sibeyn [Journal of Algorithms 2006] when $q > m$, where $m$ is the number of blocks that can be stored in main memory. We also extend it to support insertions and deletions of elements using $O(\B_m(S_q)) + O(n + q \log_B N)$ I/Os.

\subsection{Preliminaries}
\label{section:prelim}

Given an unsorted array~$\A$ of length~$n$, the \emph{median} is the element~$x$ of $\A$ such that exactly $\lceil n/2 \rceil$ elements in $\A$ are greater than or equal to~$x$. It is well-known that the median can be computed in~$O(n)$ time, and many~\cite{hoare:quickselect,blum:median,schonhage:median} have analyzed the exact constants involved. The best known result is due to Dor and Zwick~\cite{dor:median-upperbound} to obtain $2.942 + o(n)$ time.

In the external memory model, the computer is abstracted to consist of two memory levels: the internal memory of size~$M$, and the (unbounded) disk memory, which operates by reading and writing data in blocks of size~$B$. We refer to the number of items of the input by~$N$. For convenience, we define $n=N/B$ and $m=M/B$ as the number of blocks of input and memory, respectively. We make the reasonable assumption that $1 \leq B \leq M/2$. In this model, we assume that each I/O read or write is charged one unit of time, and that an internal memory operation is charged no units of time. 
To achieve the optimal sorting bound of $SortIO(N) = \Theta(n\log_m n)$ in this setting, it is necessary to make the \emph{tall cache} assumption~\cite{brodal:limits-cache-ob}: $M = \Omega(B^{1+\epsilon})$, for some constant~$\epsilon > 0$, and we will make this assumption for the remainder of the paper.

%
%

%% file: algorithm.tex
\renewcommand{\B}{\texttt{B}}
Let $\A$ be an input array of $n$ unsorted items.
We describe a simple version of our algorithm for handling selection and
search queries on array~$\A$.
We say that an element in array~$\A$ at position~$i$ is a \emph{pivot} if $\A[1\ldots i-1] < \A[i] \le \A[i+1\ldots n]$.

\textit{Bit Vector.}
Throughout all the algorithms in the paper, we maintain a bitvector~$\B$ of length~$n$ where $\B[i]=\one$ if and only if it is a pivot.

\textit{Preprocessing.}
Create a bitvector~$\B$ and set each bit to \zero.
Find the minimum and maximum elements in array~$\A$, swap them into $\A[1]$ and $\A[n]$ respectively, and set $\B[1]=\B[n]:=\one$.

\textit{Selection.}
We define the operation~$\A.\select(s)$ to refer to the selection query~$s$, which returns $\A[s]$ if $\A$ were sorted. To compute this result, if $\B[s]=\one$ then return $\A[s]$ and we are done.
If $\B[s]=\zero$, find $a<s$, $b>s$, such that $\B[a]=\B[b]=\one$ but
  $\B[a+1\ldots b-1]$ are all \zero.
Perform quickselect~\cite{hoare:quickselect} on $\A[a+1\ldots b-1]$, marking pivots found along the way in $\B$.
This gives us $\A[s]$, with $\B[s]=\one$, as desired.

\textit{Search.}
We define the operation~$\A.\search(p)$ returns the position~$j$, which
satisfies $p=\A[j]$ if $\A$ were sorted; if $p \not\in \A$, then $j$ is the number of items in $\A$ smaller than $p$.\footnote{The $\search$ operation is essentially the same as $\rank$ on the set of elements stored in the array~$A$. We call it $\search$ to avoid confusion with the $\rank$ operation defined on bitvectors in Section~\ref{section:dynamic}.}
Perform a binary search on $\A$ \textit{as if $\A$ were sorted}.
Let $i$ be the location in $\A$ we find from the search; if along
  the way we discovered endpoints for the subarray we are searching that
  were out of order, stop the search and let $i$ be the midpoint.
If $\A[i]=p$ and $\B[i]=\one$ return $i$ and we are done.
Otherwise, we have just identified the unsorted interval in $\A$ that
  contains $p$ if it is present.
Perform a selection query on this interval;
  choose which side of a pivot on which to recurse 
  based on the \textit{value} of $p$
  (instead of an array position as would be done in a normal selection query).
As above, we mark pivots in $\B$ as we go; at the end of the recursion we will discover
  the needed value~$j$.

As queries arrive, our algorithm performs the same steps that quicksort
would perform, although not necessarily in the same order. If we receive enough queries, we will, over time, perform a quicksort on array~$\A$. This also means that our recursive subproblems mimic those from quicksort.

We have assumed, up to this point, that the last item in an interval is
  used as the pivot, and a simple linear-time partition algorithm is used.
We explore using different pivot and partitioning strategies to
  obtain various complexity results for online selection and searching. As an easy consequence of more a more precise analysis to follow, we show that the time to perform $q$ select and search queries on an
  array of $n$ items is $O(n\log q + q\log n)$.
%
Now, we define terminology for this alternate analysis.

\subsection{Terminology}
\label{entropyterminology}

For now we assume all queries are selection queries, since search queries
  are selection queries with a binary search preprocessing
  phase taking $O(\log n)$ comparisons. We explicitly bound the binary search cost in our remaining results.

\paragraph{Query and Pivot Sets.}
Let $Q$ denote a sequence of $q$ selection queries, ordered by time of arrival.
Let $S_t=\{s_i\}$ denote the first $t$ queries from $Q$, sorted by position.
  We also include $s_0=1$ and $s_{t+1}=n$ in $S_t$ for convenience of notation,
  since the minimum and maximum are found during preprocessing.
Let $P_t=\{p_i\}$ denote the set of $k$ pivots found by the algorithm
  when processing $S_t$, again sorted by position.
Note that $p_1=1$, $p_k=n$, $\B[p_i]=\one$ for all $i$,
  and $S_t\subseteq P_t$.

\paragraph{Pivot Tree, Recursion Depth, and Intervals.}
The pivots chosen by the algorithm form a binary tree structure,
defined as the \emph{pivot tree}~$T$ of the algorithm over time.\footnote{Intuitively, a pivot tree corresponds to a \emph{recursion tree}, since each node represents one recursive call made during the quickselect algorithm~\cite{hoare:quickselect}.}
Pivot~$p_i$ is the parent of pivot~$p_j$ if, after $p_i$ was used to
  partition an interval, $p_j$ was the pivot used to partition either the
  right or left half of that interval.
The root pivot is the pivot used to partition $\A[2..n-1]$ due to preprocessing. The \emph{recursion depth}, $d(p_i)$, of a pivot $p_i$ is the length of the
  path in the pivot tree from~$p_i$ to the root pivot.
All leaves in the pivot tree are also selection queries, but 
  it may be the case that a query is not a leaf.
Each pivot was used to partition an interval in $\A$. 
  Let~$I(p_i)$ denote the interval partitioned by~$p_i$ (which may be empty),
  and let $|I(p_i)|$ denote its length.
Intervals form a binary tree induced by their pivots.
If $p_i$ is an ancestor of $p_j$ then $I(p_j)\subset I(p_i)$.
The recursion depth of an array element is the recursion depth of the
  smallest interval containing that element, which in turn is the
  recursion depth of its pivot.

\renewcommand{\B}{\mathcal{B}}

\paragraph{Gaps and Entropy.}
Define the query gap $\Delta_i^{S_t}:=s_{i+1}-s_i$ and similarly
the pivot gap $\Delta_i^{P_t}:=p_{i+1}-p_i$.
Observe that each pivot gap is contained in a smallest interval~$I(p)$.
  One endpoint of this gap is the pivot~$p$ of interval~$I(p)$, and the other
  matches one of the endpoints of interval~$I(p)$.
By telescoping we have $\sum_i \Delta_i^{S_t} = \sum_j \Delta_j^{P_t} = n-1$.

We will analyze the complexity of our algorithms based on the number of element comparisons.
%
The lower bound on the number of comparisons required
to answer the selection queries in $S_t$ is obtained by taking
  the number of comparisons to sort the entire array, and then
  subtracting the comparisons needed to sort the query gaps.
We use $\B(S_t)$ to denote this lower bound.
\begin{eqnarray*}
  \B(S_t) &:=& 
       \sum_{i=0}^t \left(\Delta_i^{S_t}\right) \log \left(n/\left(\Delta_i^{S_t}\right)\right) - O(n).
\end{eqnarray*}
Note that $\B(S_q) \le n\log q$: this upper bound is met when
the queries are evenly spaced over the input array~$\A$. We can show that the simple algorithm performs~$O(\B(S_q)+q\log n)$ for a sequence~$Q$ of~$q$ select and search queries on an array of~$n$ elements.
%
We will also make use of the following fact in the paper.
\begin{fact}
\label{fact:logloglog}
For all $\epsilon > 0$, there exists a constant~$c_\epsilon$ such that for all $x \geq 4$, $\log \log \log x < \epsilon \log x + c_\epsilon$.
\end{fact}
\begin{proof}
Since $\lim_{x\rightarrow\infty} (\log\log\log x)/(\log x) = 0$, there exists a $k_\epsilon$ such that for all $x \geq k_\epsilon$, we know that $(\log\log\log x)/(\log x) < \epsilon$. Also, we know that in the interval~$[4,k_\epsilon]$, the continuous function~$\log\log\log x - \epsilon\log x$ is bounded. Let $c_\epsilon = \log\log\log k_\epsilon - 2\epsilon$, which is a constant.
\end{proof}

%% file: entropy.tex
\textit{Pivot Selection Methods.}
We say that a pivot selection method is \textit{good} for
  the constant $c$ with $1/2\le c<1$ if, for all pairs of
  pivots $p_i$ and $p_j$ where $p_i$ is an ancestor of
  $p_j$ in the pivot tree, then
$$
  |I(p_j)| \le |I(p_i)| \cdot c^{d(p_j)-d(p_i)+O(1)}.
$$
Note that if the median is always chosen as the pivot, we have $c=1/2$ and the
  $O(1)$ term is in fact zero.
The pivot selection method of Kaligosi\etal~\cite[Lemma 8]{kaligosi-multiselect} is \textit{good} with $c=15/16$.

\begin{lemma}
\label{entropylemma}
  If the pivot selection method is \textit{good} as defined above,
  then $\B(P_t)=\B(S_t)+O(n)$.
\end{lemma}
\begin{proof}
We sketch the proof and defer the full details to Appendix~\ref{appendix:entropylemma}.
Consider any two consecutive selection queries $s$ and $s^\prime$,
and let $\Delta=s^\prime-s$ be the gap between them.
Let $P_\Delta=(p_l, p_{l+1}, \ldots, p_r)$ be the pivots in this gap,
where $p_l=s$ and $p_r=s^\prime$.
The lemma follows from the claim that $\B(P_\Delta)=O(\Delta)$, since
\begin{eqnarray*}
	\B(P_t) - \B(S_t) & = & \left(n \log n - \sum_{j=0}^k \Delta_j^{P_t} \log \Delta_j^{P_t}\right) - \left(n \log n - \sum_{i=0}^{t} \Delta_i^{S_t} \log \Delta_i^{S_t}\right) \\
	& = & \sum_{i=0}^{t} \Delta_i^{S_t} \log \Delta_i^{S_t} - \sum_{j=0}^k \Delta_j^{P_t} \log \Delta_j^{P_t} \\
	& = & \sum_{i=0}^t \B(P_{\Delta_i^{S_t}}) = \sum_{i=0}^t O\left(\Delta_i^{S_t}\right) = O(n). \\
\end{eqnarray*}
We now sketch the proof of our claim, which proves the lemma.

There must be a unique pivot~$p_m$ in $P_\Delta$ of minimal recursion depth.
We split the gap~$\Delta$ at~$p_m$. We define For brevity, we define $D_l=\sum_{i=0}^{m-1} \Delta_i$ and $D_r=\sum_{i=m}^{r-1} \Delta_i$, giving $\Delta=D_l+D_r$. 

We consider the proof on the right-hand side~$D_r$, and proof for $D_l$ is similar. Since we use a good pivot selection method, we can bound the total information content of the right-hand side by~$O(D_r)$. This leads to the claim, and the proof follows. Details of this proof are in Appendix~\ref{appendix:entropylemma}.
\end{proof}

\begin{theorem}[Online Multiselection]
Given an array of $n$ elements, on which we have performed a sequence~$Q$ of $q$ online selection and search queries, of which $q'$ are search, we provide
\begin{itemize}
\item a randomized online algorithm that performs the queries using
  $\B(S_q)+O(n + q'\log n)$ expected number of comparisons, and
\item a deterministic online algorithm that performs the queries
  using at most $4\B(S_q)+O(n+q'\log n)$ comparisons.
\end{itemize}
\end{theorem}
\begin{proof}
For the randomized algorithm, we use the randomized pivot selection
algorithm of Kaligosi\etal~\cite[Section 3, Lemma 2]{kaligosi-multiselect}.)
This algorithm gives a good pivot selection method with $c=1/2+o(1)$, and
the time to choose the pivot is $O(\Delta^{3/4})$ on an interval of
length $\Delta$, which is subsumed in the $O(n)$ term in the running time.
Each element in an interval participates in one comparison per partition
operation.
Thus, the total number of comparisons is expected to be the
sum of the recursion depths of all elements in the array.
This total is easily shown to be $\B(P_q)$, and by Lemma~\ref{entropylemma}, the proof is complete. In Appendix~\ref{section:randomized}, we describe how to get a good pivot selection method with just~$6(\log n)^3(\log \Delta)^2$ samples, instead of $O(\Delta^{3/4})$.

For the deterministic algorithm, we use the median of each interval as the
pivot; the median-finding algorithm of Dor and Zwick \cite{dor:median-upperbound} gives this to us in under $3\Delta$ comparisons.
We add another comparison for the partitioning, to give a count of
comparisons per array element of four times the recursion depth. This is at most $4\B(P_q)$, which is no more than $4\B(S_q) + O(n)$ from Lemma~\ref{entropylemma},
and the result follows.
\end{proof}

%% file: optimal.tex
In this section we prove the following theorem.

\begin{theorem}[Optimal Online Multiselection]
\label{theorem:onlineoptimaldeterministic}
Given an unsorted array $\A$ of $n$ elements, we provide a 
deterministic algorithm that supports a sequence~$Q$ of~$q$ 
online selection and search queries, of which~$q'$ are search, 
using $\B(S_q)(1+o(1))+O(n+q'\log n)$ comparisons in the worst case.
\end{theorem}

Note that our bounds match those of the offline algorithm of Kaligosi\etal~\cite{kaligosi-multiselect} when $q'=0$ (i.e., there are no search queries). In other words, we provide the first $1$-competitive online multiselection algorithm.
We explain our proof with three main steps.
We first explain our algorithm and how it is different from the algorithm in~\cite{kaligosi-multiselect}.
We then bound the number of comparisons from merging
  by $\B(S_q)(1+o(1))+O(n)$,
and then we bound the number of comparisons from pivot finding and
  partitioning by $o(\B(S_q))+O(n)$.

\subsection{Algorithm Description and Modifications}
We briefly describe the deterministic algorithm from Kaligosi\etal~\cite{kaligosi-multiselect}. They begin by creating \emph{runs}, which are sorted sequences from~$\A$ of length roughly~$\ell=\log(\B/n)$. Then, they compute the median~$m$ of the median of these sequences and partition the runs based on~$m$. After partitioning, they recurse on the two sets of runs, sending $\select$ queries to the appropriate side of the recursion. To maintain the invariant on run length on the recursions, they merge short like-sized runs optimally until all but~$\ell$ of the runs are again of length between~$\ell$ and~$2\ell$.

We make the following modifications to the deterministic algorithm of Kaligosi\etal~\cite{kaligosi-multiselect}:
\begin{itemize}
  \item 
    The queries are processed online, that is, one at a time,
    from~$Q$ without knowing which queries will follow.
    To do this, we maintain the bitvector~$\texttt{B}$ as described above.
  \item
    We admit search queries in addition to selection queries;
    in the analysis we treat them as selection queries, paying
    $O(q'\log n)$ comparisons to account for binary search.
  \item 
    Since we don't know all of~$Q$ at the start, we cannot know the value of
    $\B(S_q)$ in advance. Therefore, we cannot preset a value for $\ell$ as in Kaligosi\etal~\cite{kaligosi-multiselect}.
    Instead, we set $\ell$ locally in an interval $I(p)$ to
    $1+\lfloor \lg(d(p)+1) \rfloor$.
    Thus, $\ell$ starts at $1$ at the root of the pivot tree~$T$, and
    since we use only good pivots, $d(p)=O(\lg n)$. (Also, $\ell=\log\log n + O(1)$ in the worst case.)
    We keep track of the recursion depth of pivots, from which it is
      easy to compute the recursion depth of an interval.
    Also observe that $\ell$ can increase by at most one when moving down
      one recursion level during a selection.
  \item
    We use a second bitvector~$\R$ to identify the endpoints of runs
    within each interval that has not yet been partitioned.
\end{itemize}
The algorithm to perform a selection query is as follows:
\begin{itemize}
  \item
    As described earlier in this paper, we use bitvector~$\texttt{B}$ to
      identify the interval from which to begin processing.
    The minimum and maximum are found in preprocessing.
  \item
    If the current interval has length less than $4\ell^2$, we
      sort the interval to complete the query (setting all elements as pivots). The cost for this case is bounded by Lemma~\ref{medpartsmallnode}.
  \item
    As is done in~\cite{kaligosi-multiselect},
      we compute the value of $\ell$ for the current interval,
      merge runs so that there is at most one of each length $<\ell$,
      and then use medians of those runs to
      compute a median-of-medians to use as a pivot.
    We then partition each run using binary search.
\end{itemize}
We can borrow much of the analysis done in~\cite{kaligosi-multiselect}.
We cannot use their work wholesale, because we don't know $\B$ in advance. For this reason, we cannot define $\ell$ as they have, and their algorithm depends heavily on its use. To finish the proof of our theorem, we show how to modify their techniques to handle this complication.

\renewcommand{\B}{\mathcal{B}}
\subsection{Merging}
Kaligosi\etal~\cite[Lemmas~5---10]{kaligosi-multiselect} count the comparisons resulting from merging. Lemmas 5, 6, and 7 do not depend on the value of~$\ell$ and so we can use them in our analysis. Lemma 8 shows that the median-of-medians built on runs is a good pivot selection method.
Although the proof clearly uses the value of $\ell$, its validity does not
  depend on how large $\ell$ is; only that there are at least $4\ell^2$ items
  in the interval, which also holds for our algorithm.
Lemmas 9 and 10 together will bound the number of comparisons
  by $\B(S_q)(1+o(1))+O(n)$ if we can prove Lemma~\ref{lemma:run-info-theoretic}, which bounds the information content of runs in intervals that are not yet partitioned.

\begin{lemma}
\label{lemma:run-info-theoretic}
Let a run~$r$ be a sorted sequence of elements from~$\A$ in a gap~$\Delta_i^{P_t}$, where $|r|$ is its length. Then,
$$
 \sum_{i=0}^k \sum_{ r \in \Delta_i^{P_t} } |r|\lg|r| = 
   o(\B(S_t) ) + O(n).
$$
\end{lemma}
\begin{proof}
  In a gap of size $\Delta$, $\ell=O(\log d)$
    where $d$ the recursion depth of the elements in the gap.
  This gives
 $\sum_{ r \in \Delta } |r|\log|r| \le \Delta \log (2l) = O(\Delta \log\log d)$,
  since each run has size at most $2\ell$.
  Because we use a good pivot selection method, we know that the recursion depth of every element in the gap is $O(\log(n/\Delta))$.
  Thus,
 $\sum_{i=0}^k \sum_{ r \in \Delta_i^{P_t} } |r|\log|r|  \le
  \sum_i \Delta_i \log\log\log(n/\Delta_i)$.
  Recall that $\B(S_t)=\B(P_t)+O(n)=\sum_i \Delta_i \log(n/\Delta_i)+O(n)$.
  Using Fact~\ref{fact:logloglog}, the proof is complete.
\end{proof}

\subsection{Pivot Finding and Partitioning} 

\input{medianpartition}

%% file: medianpartition.tex
Now we prove that the cost of computing medians and performing partition requires at most $o(\B(S_q)) + O(n)$ comparisons. The algorithm computes the median~$m$ of medians of each run at a node~$v$ in the pivot tree~$T$. Then, it partitions each run based on~$m$. We bound the number of comparisons at each node~$v$ with more than $4\ell^2$ elements in Lemmas~\ref{medpartnode} and~\ref{medpartlevel}. We bound the comparison cost for all nodes with fewer elements in Lemma~\ref{medpartsmallnode}.

\paragraph{Terminology.} Let $d$ be the current depth of the pivot tree~$T$ (defined in Section~\ref{entropyterminology}), and let the root of~$T$ have depth~$d=0$. In tree~$T$, each node~$v$ is associated with some interval~$I(p_v)$ corresponding to some pivot~$p_v$. We define $\Delta_v = |I(p_v)|$ as the number of elements at node~$v$ in~$T$.

Recall that $\ell = 1 + \lfloor \log(d+1)\rfloor$. Let a \emph{run} be a sorted sequence of elements from~$\A$. We define a \emph{short run} as a run of length less than $\ell$. Let $\beta n$ be the number of comparisons required to compute the exact median for $n$ elements, where $\beta$ is a constant less than three~\cite{dor:median-upperbound}. Let $r_v^s$ be the number of short runs at node~$v$, and let $r_v^l$ be the number of long runs.

\begin{lemma}
\label{medpartnode}
The number of comparisons required to find the median~$m$ of medians and partition all runs at~$m$ for any node~$v$ in the pivot tree~$T$ is at most $\beta(\ell-1) + \ell \log \ell + \beta(\Delta_v/\ell) + (\Delta_v/\ell) \log (2\ell)$ comparisons.
\end{lemma}
\begin{proof}
We compute the cost (in comparisons) for computing the median of medians.
For the $r_v^s \leq \ell-1$ short runs, we need at most $\beta (\ell-1)$ comparisons per node. For the $r_v^l \leq \Delta_v/\ell$ long runs, we need at most $\beta (\Delta_v/\ell)$.

Now we compute the cost for partitioning each run based on~$m$. We perform binary search on each run. For short runs, this requires at most $\sum_{i=1}^{\ell-1} \log i \leq \ell \log \ell$ comparisons per node. For long runs, we need at most $(\Delta_v/\ell) \log (2\ell)$ comparisons per node.
\end{proof}

Since our value of~$\ell$ changes at each level of the recursion tree, we will sum the above costs by level. The overall cost in comparisons at level~$d$ is at most
\begin{displaymath}
2^d \beta\ell + 2^d \ell \log \ell + (n/\ell)\beta + (n/\ell)\log (2\ell).
\end{displaymath}
We can now prove the following lemma.

\begin{lemma}
\label{medpartlevel}
The number of comparisons required to find the median of medians and partition over all nodes~$v$ in the pivot tree~$T$ with at least $4\ell^2$ elements is at most $o(\B(S_t)) + O(n)$.
\end{lemma}
\begin{proof}
For all levels of the pivot tree up to level~$\ell' \leq \log(\B(P_t)/n)$, the cost is at most
\begin{eqnarray*}
\sum_{d=1}^{\log (\B(P_t)/n)} 2^d \ell (\beta + \log \ell) + (n/\ell)(\beta + \log (2\ell)).
\end{eqnarray*}
Since $\ell = \lfloor\log (d+1)\rfloor + 1$, the first term of the summation is bounded by $(\B(P_t)/n) \log \log (\B(P_t)/n) = o(\B(P_t))$. The second term is easily upper-bounded by 
\begin{displaymath}
n \log(\B(P_t)/n) (\log \log \log (\B(P_t)/n)/\log\log(\B(P_t)/n)) = o(\B(P_t)).
\end{displaymath}
Using Lemma~\ref{entropylemma}, the above two bounds are $o(\B(S_t)) + O(n)$.

For each level~$\ell'$ with $\log(\B(P_t)/n) < \ell' \leq \log\log n + O(1)$,
we need to bound the remaining cost. It is easy to bound each node~$v$'s cost by $o(\Delta_v)$, but this is not sufficient---though we have shown that the total number of \emph{comparisons} for merging is $\B(S_t) + O(n)$, the number of \emph{elements} in nodes with $\Delta_v \geq 4\ell^2$ could be $\omega(\B(S_t))$.

We bound the overall cost as follows, using the result of Lemma~\ref{medpartnode}. Since node~$v$ has $\Delta_v > 4 \ell^2$ elements, we can rewrite the bounds as  $O(\Delta_v/\ell \log (2\ell))$. Recall that $\ell = \log d + O(1) = \log (O(\log (n/\Delta_v))) = \log\log(n/\Delta_v) + O(1)$, since we use a good pivot selection method. Summing over all nodes, we get
$\sum_v \left(\Delta_v/\ell\right) \log (2\ell) \leq \sum_v \Delta_v \log (2\ell) = o\left(\B(P_t)\right) + O(n)$, using Fact~\ref{fact:logloglog} and recalling that $\B(P_t) = \sum_v \Delta_v \log (n/\Delta_v)$.
Finally, using Lemma~\ref{entropylemma}, we arrive at the claimed bound for queries.
\end{proof}

Now we show that the comparison cost for all nodes~$v$ where $\Delta_v \leq 4 \ell^2$ is at most $o(\B(S_t)) + O(n)$.

\begin{lemma}
\label{medpartsmallnode}
For nodes~$v$ in the pivot tree~$T$ where~$\Delta_v \leq 4 \ell^2$, the total cost in comparisons for all operations is at most $o(\B(S_t)) + O(n)$.
\end{lemma}
\begin{proof}
We observe that nodes with no more than $4\ell^2$ elements do not incur any cost in comparisons for median finding and partitioning, unless there is (at least) one associated query within the node. Hence, we focus on nodes with at least one query.

Let $z = (\log\log n)^2\log\log\log n + O(1)$. We sort the elements of any node~$v$ with $\Delta_v \leq 4\ell^2$ elements using $O(z)$ comparisons, since $\ell \leq \log\log n + O(1)$. We set each element as a pivot. The total comparison cost over all such nodes is no more than $O(tz)$, where $t$ is the number of queries we have answered so far. If $t < n/z$, then the above cost is $O(n)$.

Otherwise, $t \geq n/z$. Then, we know that $\B(P_t) \geq (n/z) \log (n/z)$, by Jensen's inequality. (In words, this represents the sort cost of $n/z$ adjacent queries.) Thus, $tz \in o(\B(P_t))$. Using Lemma~\ref{entropylemma}, we know that $\B(P_t) = \B(S_t) + O(n)$, thus proving the lemma.
\end{proof}

%% file: dynamic.tex
In this section, we extend our results for the case of the static array by allowing 
insertions and deletions of elements in the array, while supporting the selection 
queries. Recall that we are originally given the unsorted list~\A. 
For supporting $\ins$ and $\delete$ efficiently, we maintain the newly inserted elements in a
separate data structure, and mark the deleted elements in $A$. These $\ins$ and $\delete$ operations are 
occasionally merged to make the array up-to-date. 
Let $\A'$ denote the current array with length $n'$. We want to support the following 
two additional operations:
\begin{itemize}
	\item \ins($a$), which inserts $a$ into $\A'$, and;
	\item \del($i$), which deletes the $i$th (sorted) entry from~$\A'$.
\end{itemize}

\subsection{Preliminaries}
\label{section:dynamic-preliminaries}

Our solution uses the \emph{dynamic bitvector} data structure of Hon~et~al.~\cite{hon:dynamic-bitvector}. This structure supports the following set of operations on a dynamic bitvector~$\V$. 
	The $\rankb(i)$ operation tells the number of~$b$ bits up to the $i$th position in~$\V$.
	The $\selectb(i)$ operation gives the position in~$\V$ of the $i$th $b$ bit.
	The $\ins_b(i)$ operation inserts the bit~$b$ in the $i$th position.
	The $\delete(i)$ operation deletes the bit located in the $i$th position.
	The $\flip(i)$ operation flips the bit in the $i$th position.

Note that one can determine the $i$th bit of~$\V$ by computing $\rankone(i) - \rankone(i-1)$. (For convenience, we assume that $\rankb(-1) = 0$.)
The result of Hon et al.~\cite[Theorem~1]{hon:dynamic-bitvector} can be re-stated as follows, for the case of maintaining a dynamic bit vector (the result of~\cite{hon:dynamic-bitvector} is stated
for a more general case).

\begin{lemma}[\cite{hon:dynamic-bitvector}]
\label{lem:dynamic-bitvector}
Given a bitvector~$\V$ of length $n$, there exists a data structure that takes $n + o(n)$ bits and 
supports $\rankb$ and $\selectb$ in $O(\log_t n)$ time, and $\ins$, $\delete$ and $\flip$ in
$O(t)$ time, for any parameter $t$ such that $(\log n)^{O(1)} \le t \le n$. The data structure 
assumes access to a precomputed table of size $n^{\epsilon}$, for any fixed $\epsilon > 0$.
\end{lemma}

The elements in the array~$\A$ swapped during the queries and $\ins$ and $\delete$ operations,
to create new pivots, and the positions of these pivots are maintained
as before using the bitvector~$\B$. In addition, we also maintain two  
bitvectors, each of length $n'$: (i) an \emph{insert bitvector}~$\I$ such 
that $\I[i] = \one$ if and only if $\A'[i]$ is newly inserted, and (ii) a 
\emph{delete bitvector}~$\D$ such that if $\D[i] = \one$, the $i$th 
element in $\A$ has been deleted. If a newly inserted item is deleted, 
it is removed from~$\I$ directly. Both $\I$ and $\D$ are implemented 
as instances of the data structure of Lemma~\ref{lem:dynamic-bitvector}.

We maintain the values of the newly inserted elements in a balanced binary search tree~$T$. The  inorder traversal of the nodes of~$T$ corresponds to the increasing order of their positions in array $\A'$. 
We support the following operations on this tree are: (i) given an index~$i$, return the element corresponding to the $i$th node in the inorder traversal of~$T$,
and (ii) insert/delete an element at a given inorder position.
By maintaining the subtree sizes of the nodes in~$T$, these operations
 can be performed in $O(\log n)$ time 
without having to perform any comparisons between the elements.

Our preprocessing steps are the same as in the static case. In addition, the bitvectors~$I$ 
and~$D$ are each initialized to the bitvector of $n$ \zero s. The tree~$T$ is initially empty.

In addition, after performing $|\A|$ $\ins$ and $\delete$ operations, we merge all the elements 
in $T$ with the array~$\A$, modify the bitvector~$\texttt{B}$ appropriately, and 
reset the bitvectors $\I$ and $\D$ (with all zeroes). This increases the 
amortized cost of the $\ins$ and $\delete$ operations by $O(1)$, without requiring any additional
comparisons.

\subsection{Dynamic Online Multiselection}
\label{subsec:internal-dynamic}

We now describe how to support $\A'.\ins(a)$, $\A'.\del(i)$, $\A'.\select(i)$, and $\A'.\search(a)$ 
operations.

\paragraph{$\A'.\ins(a)$.} First, we search for the appropriate unsorted interval~$[\ell,r]$ containing~$a$ using a binary search on the original (unsorted) array~$\A$.
Now perform $\A.\search(a)$ on interval~$[\ell,r]$ (choosing which subinterval 
to expand based on the insertion key~$a$) until $a$'s exact position~$j$ in 
$\A$ is determined. The original array~$\A$ must have chosen as pivots the 
elements immediately to its left and right (positions~$j-1$ and $j$ in array~$\A$); 
hence, one never needs to consider newly-inserted pivots when choosing subintervals.
Insert~$a$ in sorted order in~$T$ among at position $\I.\selectone(j)$ among 
all the newly-inserted elements.  Calculate $j'=\I.\selectzero(j)$, and set $a$'s 
position to $j''=j'-\D.\rankone(j')$. Finally, we update our bitvectors by 
performing $\I.\ins_\one(j'')$ and $\D.\ins_\zero(j'')$.
Note that, apart from the $\search$ operation, all other operations in the 
insertion procedure do not perform any comparisons between the elements.

\paragraph{$\A'.\delete(i)$.} First compute~$i'=\D.\selectzero(i)$. 
If $i'$ is newly-inserted (i.e., $\I[i'] = \one$), then remove the node 
(element) with inorder number~$\I.\rankone(i')$ from $T$.
Then perform $\I.\delete(i')$ and $\D.\delete(i')$. 
If instead $i'$ is an older entry, 
simply perform $\D.\flip(i')$. In other words, 
we mark the position $i'$ in $A$ as deleted even though the corresponding 
element may not be in its proper place.\footnote{If a user wants to 
delete an item with value~$a$, one could simply search for it first to 
discover its rank, and then delete it using this function.}

\paragraph{$\A'.\select(i)$.} If $\I[i] = \one$, return the element 
corresponding to the node with inorder number $\I.\rankone(i)$ 
in $T$. Otherwise, compute $i'=\I.\rankzero(i)-\D.\rankone(i)$, 
and return $\A.\select(i')$).

\paragraph{$\A'.\search(a)$.} First, we search for the appropriate 
unsorted interval~$[\ell,r]$ containing~$a$ using a binary search 
on the original (unsorted) array~$\A$. Then, perform $\A.\search(a)$ 
on interval~$[\ell,r]$ until $a$'s exact position~$j$ is found. If $a$ 
appears in array~$\A$ (which we discover through \search), we 
need to now check whether it has been deleted. We compute 
$j' = \I.\selectzero(j)$ and $j'' = j' - \D.\rankone(j')$. If $\D[j'] = \zero$, 
return $j''$. Otherwise, it is possible that the item has been 
newly-inserted. Compute $p=\I.\rankone(j')$, which is the number of 
newly-inserted elements that are less than or equal to $a$. 
If $T[p]=a$, then return $j''$; otherwise, return failure.

We show that the above algorithm achieves the following performance 
(in Appendix~\ref{appendix:dynamic}).

\begin{theorem}[Optimal Online Dynamic Multiselection]
\label{thm:dynamic-internal}
Given a dynamic array~$\A'$ of $n$ original elements, there exists a dynamic online data structure that can support $q = O(n)$ $\select$, $\search$, $\ins$, and $\delete$ operations, of which $q'$ are $\search$, $\ins$, and $\delete$, we provide a deterministic online algorithm that uses at most $\B(S_q) (1 + o(1)) + O(n + q' \log n)$ comparisons.
\end{theorem}

%% file: EXT.tex
\label{tempsec:online-select-rank-1}

Suppose we are given an unsorted array $\A$ of length $N$ stored in $n = N/B$ blocks in
the external memory. Recall that sorting $\A$ in the external memory model requires $SortIO(N) = \Theta(n\log_m n)$ I/Os.
The techniques we use in main memory are not immediately applicable to the external memory model. In the extreme case where we have $q=N$ queries, the internal memory solution would require~$O(n\log_2 (n/m))$ I/Os. This compares poorly to the optimal~$O(n\log_m n)$ I/Os performed by the optimal mergesort algorithm for external memory.


As in the case of internal memory, the lower bound on the 
number of I/Os required to perform a given set of selection queries
can be obtained by subtracting the number of I/Os required to 
sort the elements between the `query gaps' from the sorting bound.
More specifically, let $S_t=\{s_i\}$ be the first $t$ queries from a
query set $Q$, sorted by position, and for $1 \le i \le t$, let
$\Delta_i^{S_t}:=s_{i+1}-s_i$ be the query gaps, as defined 
in~Section~\ref{entropyterminology}. Then the lower bound on the 
number of I/Os required to support the queries in $S_t$ is given by
\begin{eqnarray*}
  \B_m(S_t) &:=& 
   n \log_m n - \sum_{i=0}^{t} \left(\Delta_i^{S_t}/B\right) \log_m \left(\Delta_i^{S_t}/B\right) - O(n), \\
\end{eqnarray*}

where we assume that $\log_m \left(\Delta_i^{S_t}/B\right) = 0$ when $\Delta_i^{S_t} < mB =M$ in the above definition.

\subsection{ Algorithm Achieving $O(\B_m(S_q)) + O(n)$ I/Os}
\label{tempsec:an-algorithm-online}

We now show that our lower bound is asymptotically tight, by 
describing an $O(1)$-competitive algorithm. We assume that $\log N = \log n + \log B = O(B)$---which allows us to store a pointer to a block of the input using a constant number of blocks. This constraint is a reasonable assumption in practice, and is similar to the word-size assumption transdichotomous word RAM model~\cite{FredmanW93}. In addition, the algorithm of Sibeyn~\cite{sibeyn:external-selection} only works under this assumption, though this is not explicitly mentioned.
We obtain the following result for the external memory model.
\begin{theorem}
\label{tempthm:EXTOnline}
Given an unsorted array~$\A$ occupying $n$ blocks in external memory, we provide a deterministic algorithm that supports a sequence~$Q$ of~$q$ online selection queries using $O(\B_m(S_q)) + O(n)$ I/Os under the condition that $\log N = O(B)$.
\end{theorem}
\begin{proof}
Our algorithm uses the same approach as the internal memory algorithm, except that it chooses $d-1$ pivots at once using Lemma~\ref{lma:EXTdMedian}. Hence, each node~$v$ of the pivot tree~$T$ containing~$\Delta_v$ elements has a branching factor of $d$. It subdivides its $\Delta_v$ elements into $d$ partitions. Using Lemma~\ref{lma:EXTPartition}, we know this requires~$2\delta_v + d$ I/Os, where $\delta_v = \Delta_v/B$.

We choose $d = m/2$, which satisfies the constraints for Lemmas~\ref{lma:EXTdMedian}---\ref{lma:EXTPartition}. We also maintain the bitvector~$\V$ of length~$N$, as described before. For each $\A.\select(i)$ query, we access position~$\V[i]$. If $\V[i] = \one$, return $\A[i]$, else scan left and right from the $i$th position to find the endpoints of this interval~$I_i$ using $|I_i|/B$ I/Os. The analysis follows directly from the internal algorithm.
\end{proof}

We extend this result to also support $\search$, $\ins$, and $\delete$ operations in Appendix~\ref{appendix:online-select-rank-1}.

%% file: appendix.tex
\section{Randomized Algorithm}
\label{section:randomized}

\input{random}


\section{Proof of Lemma~\ref{entropylemma} (Entropy Lemma)}
\label{appendix:entropylemma}

\begin{proof}
Consider any two consecutive selection queries $s$ and $s^\prime$,
and let $\Delta=s^\prime-s$ be the gap between them.
Let $P_\Delta=(p_l, p_{l+1}, \ldots, p_r)$ be the pivots in this gap,
where $p_l=s$ and $p_r=s^\prime$.
The lemma follows from the claim that $\B(P_\Delta)=O(\Delta)$, since
\begin{eqnarray*}
	\B(P_t) - \B(S_t) & = & \left(n \log n - \sum_{j=0}^k \Delta_j^{P_t} \log \Delta_j^{P_t}\right) - \left(n \log n - \sum_{i=0}^{t} \Delta_i^{S_t} \log \Delta_i^{S_t}\right) \\
	& = & \sum_{i=0}^{t} \Delta_i^{S_t} \log \Delta_i^{S_t} - \sum_{j=0}^k \Delta_j^{P_t} \log \Delta_j^{P_t} \\
	& = & \sum_{i=0}^t \B(P_{\Delta_i^{S_t}}) = \sum_{i=0}^t O\left(\Delta_i^{S_t}\right) = O(n). \\
\end{eqnarray*}
We now proceed to prove our claim.

There must be a unique pivot in $P_\Delta$ of minimal recursion depth.
Any pair of pivots with the same recursion depth must have a common ancestor,
and this ancestor must lie between the pair. This ancestor is in $P_\Delta$
  and it has smaller recursion depth than the pair.
Let $p_m$ denote the pivot of minimum depth.
(Note that $p_m=s$ or $p_m=s^\prime$ are possible.)
As before, define the gaps $\Delta_i=p_{i+1}-p_i$ for $l\le i < r$.
We split the gap $\Delta$ at $p_m$.
We address the right side first, and the argument for the left side is
similar.

The sequence $d(p_m), d(p_{m+1}), \ldots, d(p_{r-1})$ must
  be strictly increasing. Otherwise, one of these pivots
  must be a leaf in the pivot tree, and hence a query, which is a contradiction.

Now consider $I(p_{m+1})$.
This interval must have $p_m$ as its left endpoint, due to its smaller
  recursion depth.
Its right endpoint must have recursion depth shallower than $p_{m+1}$,
  and hence it contains all pivots up to and including $p_r$.
This means 
  $I(p_i)\subset I(p_{m+1})$ for $m+1<i<r$,
  $\Delta_i = p_{i+1}-p_i < |I(p_i)|$ for $m+1<i<r$,
  and it means that $\Delta \le |I(p_{m-1})|+|I(p_{m+1})|$.

For brevity, we define 
      $D_l=\sum_{i=0}^{m-1} \Delta_i$ and
      $D_r=\sum_{i=m}^{r-1} \Delta_i$, giving $\Delta=D_l+D_r$,
  and further $D_l \le |I(p_{m-1})|$,
   $D_r\le|I(p_{m+1})|$.
Let us also define
  $\alpha_i := D_r/\Delta_i$ for $m\le i < r$.
We have
\begin{eqnarray*}
  D_r\log D_r - \sum_{i=m}^{r-1} \Delta_i \log \Delta_i
     = \sum_{i=m}^{r-1} \Delta_i \log (D_r/\Delta_i)
      = D_r \sum_{i=m}^{r-1}  (\log \alpha_i)/\alpha_i.
\end{eqnarray*}
This quantity can be bounded from above with a lower bound on $\alpha_i$.
Write $D_r= b\cdot |I(p_{m+1})|$ for a constant $b$ with $0<b\le 1$.
So we have
\begin{eqnarray*}
  \alpha_i 
  = D_r/\Delta_i
  > D_r/|I(p_i)|
  = b |I(p_{m+1})| / |I(p_i)|.
\end{eqnarray*}
Since we are using a \emph{good} pivot selection method, we get the bound
$$
  |I(p_i)| \le |I(p_{m+1})| \cdot c^{d(p_i)-d(p_{m+1})+O(1)}.
$$
Plugging in gives us $\alpha_i > b \cdot c^{-d(p_i)+d(p_{m+1})+O(1)}
  \ge b \cdot c^{m+1-i+O(1)} $.
The last inequality used the fact that the recursion depths must be strictly increasing.
Then
     $$ \sum_{i=m}^{r-1} \frac{ \log \alpha_i}{\alpha_i}
    \le
      \sum_{j=0}^{r-1-m} \frac{ \log (b c^{j+O(1)})}{b c^{j+O(1)}} = O(1).$$
And thus
  $$D_r\log D_r - \sum_{i=m}^{r-1} \Delta_i \log \Delta_i
    = O(D_r).$$
A similar argument on the left side gives
  $$D_l\log D_l - \sum_{i=0}^{m-1} \Delta_i \log \Delta_i
    = O(D_l).$$

Finally, $\Delta\log\Delta - D_r\log D_r - D_l\log D_l=O(\Delta)$,
  and the proof is complete.
\end{proof}


\section{Proof of Theorem~\ref{thm:dynamic-internal}}
\label{appendix:dynamic}

Let $\A'$ denote the current array of length~$n'$, after a 
sequence of queries and insertions. 
Let~$Q$ be the sequence of $q$ selection operations performed 
(either directly or indirectly through other operations) on~$\A'$, ordered 
by time of arrival. Let $S_q$ be the queries of $Q$, ordered by 
position.
We now analyze the number of comparisons performed by a 
sequence of queries and $\ins$ and $\delete$ operations. 

We consider the case when 
the number of $\ins$ and $\delete$ operations is less than~$n$. In other words, we are 
between two rebuildings of our dynamic data structure. If $q'$ is 
the number of $\search$, $\ins$, and $\delete$ operations in the sequence, 
then we perform $O(q' \log n')$ comparisons to perform the 
required searches. Note that our algorithm does not perform any comparisons for $\delete(i)$ operations, until some other query is in the same interval as~$i$. The deleted element will participate in the other costs (merging, pivot-finding, and partitioning) for these other queries, but its contribution can be bounded by $O(\log n)$, which we have as a credit.

Since a $\delete$ operation does not perform any additional 
comparisons beyond those needed to perform a $\search$, 
we assume that all the updates are insertions in the rest of this
section.
Since each inserted element becomes a pivot immediately, it does not 
contribute to the comparison cost of any other $\select$ operation.
Also, note that in the algorithm of 
Theorem~\ref{theorem:onlineoptimaldeterministic}, no pivot is part of a 
run and hence cannot effect the choice of any future pivot.

Since $Q$ is essentially a set of $q$ selection queries, we can bound 
its total comparison cost for selection queries by 
Theorem~\ref{theorem:onlineoptimaldeterministic},
which gives a bound of $\B(S_q)(1+o(1))+O(n)$. This proves the theorem.


\section{External Online Multiselection}
\label{appendix:online-select-rank-1}

\input{EXTappendix}

%% file: random.tex
Our pivot-choosing method is simple and randomized.
We choose $2m$ elements at random from an interval of size $\Delta$,
sort them (or use a median-finding algorithm) to find the median,
and use that for our pivot.
We wish to set values of $m$ and $t$ such that two events happen:
\begin{itemize}
  \item At least $2t$ elements are chosen in an interval of
    size $2\Delta/\log \Delta$ about the median of the interval.
  \item Between $m-t$ and $m+t$ elements are chosen less than the median.
  \item Between $m-t$ and $m+t$ elements are chosen larger than the median.
\end{itemize}
If we can show that all events happen with probability $1-O(1/n^2)$,
then we end up with the median of our $2m$ elements being a pivot
at position $1/2(1+O(1/\log \Delta))$, which is a good pivot.

Note that the last two events are mirror images of one another, and so
have the same probability of occurring.

\textit{First Event.}
This is the simpler of the two to estimate.
A randomly chosen element fails to land in the middle interval
with probability $1-2/\log\Delta = \exp[ -2/\log\Delta (1+o(1)) ]$.
If we choose at least $(1.1)\log\Delta\log n$ elements, all fail to land in this
middle interval with probability
$(1-2/\log\Delta)^{(1.1)\log\Delta\log n}
   = \exp[ -(2.2)\log n (1+o(1)) ] = O(1/n^2)$.
Since we need $2t$ elements in the interval, it suffices for 
  $2m\ge(2.2)t\log\Delta\log n$, or $m\ge(1.1)t\log\Delta\log n$.

\textit{Second (and third) Event.}
We need a bound on the sum of the first $k$ binomial coefficients.

\input{binom}

Since choosing an element from an interval at random and observing if
it falls before or after the median is an event of probability $1/2$,
the event of choosing $2m$ elements and having less than $m-t$ fall
below the median
occurs with probability at most
$$
2^{-2m} \sum_{i=0}^{m-t-1} \binom{2m}{i}.
$$
By our lemma above, this is bounded by
  $(1/2) \exp[ -t^2/(m+t) ]$.
Thus, the probability there are between $m-t$ and $m+t$ elements
below the median is at least
$1- \exp[ -t^2/(m+t) ]$ by the symmetry of Pascal's triangle.
To obtain $1-O(1/n^2)$ we need $t^2/(m+t) > 2\log n$,
  or $t\ge\sqrt{2m\log n}(1+o(1)$.

Using our lower bound for $m$ in terms of $t$ above, we conclude that
$m=6(\log n)^3 (\log \Delta)^2$ and
$t=4(\log n)^2 \log \Delta$ meet our needs.

\begin{theorem}
Given a list of elements of length 
  $\Delta<n$, with $\Delta$ at least $6(\log n)^3(\log \Delta)^2$, 
with probability at least $1-O(1/n^2)$,
if we sample $6(\log n)^3 (\log \Delta)^2$ of the $\Delta$ elements uniformly
at random, then median of the sample falls in position
$\Delta/2 \pm \Delta/\log \Delta$ in the original list.
\end{theorem}

%% file: binom.tex
The following bound and proof are attributed to Lovasz:
\begin{lemma}
Let $0\le k < m$ and define $c:=\binom{2m}{k+1}/\binom{2m}{m}$.  Then
$$
\sum_{i=0}^k \binom{2m}{i}
  < \frac{c}{2} \cdot 2^{2m}.
$$
\end{lemma}
\begin{proof}
Write $k+1=m-t$.
Define
\begin{eqnarray*}
A&:=& \sum_{i=0}^{m-t-1} \binom{2m}{i} \\
B&:=& \sum_{i=m-t}^{m} \binom{2m}{i} \\
\end{eqnarray*}
By the definition of $c$ we have
$$
\binom{2m}{m-t}=c\binom{2m}{m}
$$
and, because the growth rate of one binomial coefficient to the next
slows as we approach $\binom{2m}{m}$, we have
$$
\binom{2m}{m-t-1}<c\binom{2m}{m-1}
$$
and thus
$$
\binom{2m}{m-t-j}<c\binom{2m}{m-j}
$$
for $0\le j\le m-t$.

Thus it follows that the sum of any $t$ consecutive binomial coefficients
is less than $c$ times the sum of the next $t$ coefficients
as long as we stay on the left-hand side of Pascal's triangle.
Thus
$A < cB + c^2B + c^3B + \cdots < \frac{c}{1-c}B$.
We also have $A+B\le 2^{2m-1}$.
Combining these we have
$$
A < \frac{c}{1-c}B \le \frac{c}{c-1}\left(2^{2m-1}-A\right).
$$
Solving for $A$ completes the proof.
\end{proof}
We then bound 
$$
\frac{ \binom{2m}{m-t} }{ \binom{2m}{m} } \le e^{-t^2/(m+t)}.
$$
This can be derived from Stirling's formula and Taylor series
estimates for the exponential and logarithm functions.
We then obtain that
\begin{lemma}
Let $0\le t < m$.  Then
$$
\sum_{i=0}^{m-t-1} \binom{2m}{i}
  <  2^{2m-1}\cdot e^{-t^2/(m+t)}.
$$
\end{lemma}


%% file: EXTappendix.tex
\label{sec:online-select-rank-1}

Suppose we are given an unsorted array $\A$ of length $N$ stored in $n = N/B$ blocks in
the external memory. Recall that sorting $\A$ in the external memory model requires $SortIO(N) = \Theta(n\log_m n)$ I/Os.
The techniques we use in main memory are not immediately applicable to the external memory model. In the extreme case where we have $q=N$ queries, the internal memory solution would require~$O(n\log_2 (n/m))$ I/Os. This compares poorly to the optimal~$O(n\log_m n)$ I/Os performed by the optimal mergesort algorithm for external memory.

\subsection{A Lower Bound for Multiselect in External Memory}
\label{sec:lower-bound-mult}

As in the case of internal memory, the lower bound on the 
number of I/Os required to perform a given set of selection queries
can be obtained by subtracting the number of I/Os required to 
sort the elements between the `query gaps' from the sorting bound.
More specifically, let $S_t=\{s_i\}$ be the first $t$ queries from a
query set $Q$, sorted by position, and for $1 \le i \le t$, let
$\Delta_i^{S_t}:=s_{i+1}-s_i$ be the query gaps, as defined 
in~Section~\ref{entropyterminology}. Then the lower bound on the 
number of I/Os required to support the queries in $S_t$ is given by
\begin{eqnarray*}
  \B_m(S_t) &:=& 
   n \log_m n - \sum_{i=0}^{t} \left(\Delta_i^{S_t}/B\right) \log_m \left(\Delta_i^{S_t}/B\right) - O(n), \\
\end{eqnarray*}

where we assume that $\log_m \left(\Delta_i^{S_t}/B\right) = 0$ when $\Delta_i^{S_t} < mB =M$ in the above definition.

\subsection{Partitioning in External Memory}
\label{sec:equiv-medi-extern}

The main difference between our algorithms for internal and
external memory is the partitioning procedure. In the internal 
memory algorithm, we partition the values according to a single pivot, recursing on the half that contains the answer. 
o
In the external memory algorithm, we modify this binary partition to a $d$-way
partition, for some $d = \Theta(m)$, by finding a sample of
$d$ ``roughly equidistant elements.''
The next two lemmas describe how to find such a sample, and 
then partition the range of values into $d+1$ subranges with 
respect to the sample. 

\begin{lemma}
\label{lma:EXTdMedian}
  Given an unsorted array~$\A$ containing~$N$ elements in 
  external memory and an integer parameter $d < M/\log N$,
  %
  one can compute a sample of 
  size $d$ from~$\A$ using $n = N/B$ I/Os, such that the rank of the 
  $j$th value in this sample is within $[j(N/d)-d, (j + \log (N/d) - 1)N/d]$.
\end{lemma}
\begin{proof}
Given an unordered  sequence of~$N$ elements (stored in a 
read-only memory) and an additional working space of size 
$S = \Omega((\log N)^2)$, Munro and Paterson~\cite{MunroPaterson}
showed how to compute a ``reasonably well-spaced sample'' of 
size $s \le S/\log N$, in a single sequential scan over the sequence.
This well-spaced sample has the property that the rank of the 
$j$th element of the sample among the initial sequence of 
elements is between $(j N / s) -1$ and $(j+ \log(N/s)-1)N/s$.

The algorithm only reads the input sequence from left to right,
and it does not perform any random accesses to the sequence. 
Note that we have access to an unbounded working space in the external memory
at the cost of additional I/Os. 
If $M = o((\log N)^2)$, we can use the space on the disk
as temporary working space.
Hence, it is easy to see that this algorithm can be translated to the
external memory model with $S = \max\{M, (\log N)^2 \}$, which gives the result stated.
\end{proof}


\begin{lemma}
\label{lma:EXTPartition}
  Given an unsorted array~$\A$ occupying $n$ pages of external memory
  and $d < M/(2B)$ sample elements stored in main memory, there is an algorithm to partition~$\A$ by those values in $2n+d$ I/Os.  
\end{lemma}

\begin{proof}
  The algorithm scans the data, keeping one input block and $d+1$
  output blocks in main memory. An output block is written to 
  external memory when it is full, or when the scan is complete.
  The algorithm performs $n$ I/O to read the input, and at most 
  $n+d+1$ I/Os to write the output into $d+1$ partitions.
\end{proof}

\subsection{ Algorithm Achieving $O(\B_m(S_q)) + O(n)$ I/Os}
\label{sec:an-algorithm-online}

We now show that our lower bound is asymptotically tight, by 
describing an $O(1)$-competitive algorithm. We assume that $\log N = \log n + \log B = O(B)$---which allows us to store a pointer to a block of the input using a constant number of blocks. This constraint is a reasonable assumption in practice, and is similar to the word-size assumption transdichotomous word RAM model~\cite{FredmanW93}. In addition, the algorithm of Sibeyn~\cite{sibeyn:external-selection} only works under this assumption, though this is not explicitly mentioned.

\begin{theorem}
\label{thm:EXTOnline}
Given an unsorted array~$\A$ occupying $n$ blocks in external memory, we provide a deterministic algorithm that supports a sequence~$Q$ of~$q$ online selection queries using $O(\B_m(S_q)) + O(n)$ I/Os under the condition that $\log N = O(B)$.
\end{theorem}
\begin{proof}
Our algorithm uses the same approach as the internal memory algorithm, except that it chooses $d-1$ pivots at once using Lemma~\ref{lma:EXTdMedian}. Hence, each node~$v$ of the pivot tree~$T$ containing~$\Delta_v$ elements has a branching factor of $d$. It subdivides its $\Delta_v$ elements into $d$ partitions. Using Lemma~\ref{lma:EXTPartition}, we know this requires~$2\delta_v + d$ I/Os, where $\delta_v = \Delta_v/B$.

We choose $d = m/2$, which satisfies the constraints for Lemmas~\ref{lma:EXTdMedian}---\ref{lma:EXTPartition}. We also maintain the bitvector~$\V$ of length~$N$, as described before. For each $\A.\select(i)$ query, we access position~$\V[i]$. If $\V[i] = \one$, return $\A[i]$, else scan left and right from the $i$th position to find the endpoints of this interval~$I_i$ using $|I_i|/B$ I/Os. The analysis follows directly from the internal algorithm.
\end{proof}

To add searches, we cannot afford to spend $\log n$ time performing binary search on the blocks of~$\texttt{B}$. To handle this case, we build a B-tree~$T$ maintaining all pivots from~$\A$. (During preprocessing, we insert~$\A[1]$ and~$\A[n]$ into~$T$.) The B-tree~$T$ will be used to support $\search$ queries in $O(\log_B N)$ I/Os instead of $O(\log N)$ I/Os. We modify the proof of Theorem~\ref{thm:EXTOnline} 
to obtain the following:

\begin{corollary}
\label{corollary:external-select-search}
Given an unsorted array~$\A$ occupying $n$ blocks in external memory, we provide a deterministic algorithm that supports a sequence~$Q$ of~$q$ online selection and search queries using $O(\B_m(S_q)) + O(\min\{qm,N\}\log_B N) + O(n)$ I/Os under the condition that $\log N = O(B)$.
\end{corollary}

Combining the ideas from Corollary~\ref{corollary:external-select-search} and Theorem~\ref{thm:dynamic-internal}, we can dynamize the above algorithm.
\begin{corollary}
\label{corollary:externaldynamic}
Given an unsorted array~$\A$ occupying $n$ blocks in external memory, we provide a deterministic algorithm that supports a sequence~$Q$ of~$q$ online \select, \search, \ins, and \delete operations using $O(\B_m(S_q)) + O(\min\{qm,N\}\log_B N) + O(n)$ I/Os under the condition that $\log N = O(B)$.
\end{corollary}